\documentclass[english]{article}

\usepackage[T1]{fontenc}
\usepackage[utf8]{inputenc}
\usepackage[english]{babel}

\usepackage{amsmath, amsthm, amssymb}

\usepackage[letterpaper]{geometry}
\usepackage{setspace}
\usepackage{titlesec}
\titleformat*{\section}{\LARGE\bfseries}
\titleformat*{\subsection}{\Large\bfseries}
\titleformat*{\subsubsection}{\large\bfseries}

\usepackage{float}
\usepackage{tikz}
\usetikzlibrary{arrows.meta}
\usepackage{enumitem}
\usepackage{array}
\usepackage{changepage}
\usepackage{comment}
\usepackage{graphicx}
\usepackage{subcaption}

\usepackage[authoryear]{natbib}

\usepackage[colorlinks=true, allcolors=blue, bookmarks=false, breaklinks=false, pdfborder={0 0 1}]{hyperref}

\providecommand{\algorithmname}{Algorithm}
\providecommand{\claimname}{Claim}
\providecommand{\corollaryname}{Corollary}
\providecommand{\definitionname}{Definition}
\providecommand{\examplename}{Example}
\providecommand{\factname}{Fact}
\providecommand{\lemmaname}{Lemma}
\providecommand{\propositionname}{Proposition}
\providecommand{\theoremname}{Theorem}

\theoremstyle{definition}
\newtheorem{example}{\examplename}
\newtheorem{defn}{\definitionname}
\theoremstyle{plain}
\newtheorem{lyxalgorithm}{\algorithmname}
\newtheorem{thm}{\theoremname}
\newtheorem{cor}{\corollaryname}
\newtheorem{prop}{\propositionname}
\newtheorem{fact}{\factname}
\newtheorem{lem}{\lemmaname}
\theoremstyle{remark}
\newtheorem{claim}{\claimname}

\providecommand{\tabularnewline}{\\}
\def\rcal{\mathcal{R}}
\def\hcal{\mathcal{H}}
\def\ttc{$\text{TTC}_{\succ}$ }

\providecommand{\algorithmname}{Algorithm}
\providecommand{\claimname}{Claim}
\providecommand{\corollaryname}{Corollary}
\providecommand{\definitionname}{Definition}
\providecommand{\examplename}{Example}
\providecommand{\factname}{Fact}
\providecommand{\lemmaname}{Lemma}
\providecommand{\propositionname}{Proposition}
\providecommand{\theoremname}{Theorem}

\newcommand{\blfootnote}[1]{%
\begingroup
\renewcommand\thefootnote{}\footnote{#1}%
\addtocounter{footnote}{-1}%
\endgroup
}

\makeatother

\begin{document}


\title{Shapley-Scarf Markets with Objective Indifferences}

\author{ Will Sandholtz\thanks{UC Berkeley. willsandholtz@econ.berkeley.edu} \and Andrew Tai\thanks{US Department of War. andrew.a.tai@gmail.com. Disclaimer: The views expressed are those of the author and do not reflect the official policy or position of the Department of War or the U.S. Government.} } \date{January 2026}

\maketitle


\begin{abstract}
Top trading cycles with fixed tie-breaking (TTC) has been suggested to deal with indifferences in object allocation problems. Unfortunately, under general indifferences, TTC is neither Pareto efficient nor group strategy-proof. Furthermore, it may not select an allocation in the core of the market, even when the core is non-empty. However, when indifferences are agreed upon by all agents (``objective indifferences''), TTC maintains Pareto efficiency, group strategy-proofness, and core selection. Further, we characterize objective indifferences as the most general setting where TTC maintains these properties.

\end{abstract}

\blfootnote{We are grateful to Haluk Ergin and Chris Shannon for guidance. We also thank seminar attendants at UC Berkeley, the University of Western Australia, the Midwest Theory Conference, and the Berlin Micro Theory and Behavioral Economics PhD Conference for helpful comments; in particular, Federico Echenique, Ivan Balbuzanov, and Yuichiro Kamada. We also thank anonymous referees for their suggestions.}


\section{Introduction}

Important markets including living donor organ transplants, public housing assignments, and school choice can be modeled as Shapley-Scarf markets: each agent is endowed with an indivisible object and has preferences over the set of objects. Monetary transfers are not allowed, and participants have property rights to their own endowments. The goal is to re-allocate these objects among the agents to achieve efficiency and stability. The usual stability notion is the core: an allocation is in the core if no subset of agents would prefer to trade their endowments among themselves. In the original setting of \citet{SS74}, agents have strict preferences over the houses, and Gale's \textit{top trading cycles} (TTC) algorithm finds an allocation in the core. \citet{RP77} further show that the core is non-empty, unique, and Pareto efficient. \citet{Roth82} shows that TTC is strategy-proof; \citet{Moulin95}, \citet{Bird84}, \citet{ST24}, and \citet{Papai00} show that it is group strategy-proof. These properties make TTC an attractive algorithm for practical applications.

However, the assumption that preferences are strict is strong. In particular, if any objects are essentially identical, agents should naturally be indifferent between them. Consider the problem of assigning students to college dormitories. It seems reasonable to assume that two units with the same floor plan in the same building are equivalent in the eyes of a student. For example, the undergraduate dorm application process at UC Berkeley applies this same logic. On housing applications, students rank their five most preferred $\textit{housing complex}\times\textit{floor plan}$ pairs. It is implicitly assumed that students have strict preferences over $\textit{housing complex}\times\textit{floor plan}$ pairs, but are indifferent between dorm units of the same type. Beyond college dormitory assignment problems, there are a host of real-world object assignment problems (military occupational specialty assignment, school choice, class assignment, etc.) that can be modeled with a similar structure.

Motivated by these examples, we study a model of Shapley-Scarf markets where there may be indistinguishable copies (which we will call ``houses'') of the objects (which we will call ``types'' or ``house types''). Our model restricts agents to be indifferent between houses of the same type, but never indifferent between houses of different types. We call these preferences ``objective indifferences.'' Objective indifferences is a minimal model of indifferences, capturing the most basic and plausible form of indifferences.

In the fully general setting where agents' preferences may contain indifferences, \textit{TTC with fixed tie-breaking} has been proposed as an intuitive extension. Ties in the agents' preference rankings are broken by some external rule, then TTC is executed on the resulting strict preference rankings. \citet{Ehlers14} is the first to formalize this procedure for Shapley-Scarf markets, though similar ideas date to earlier work. For example, \citet{AS03} deal with tie-breaking between students in the school choice setting via an exogenous priority rule. However, TTC with fixed tie-breaking has several drawbacks. First, it is not Pareto efficient nor group strategy-proof. In fact, there is an inherent tension between these two properties: \citet{Ehlers02} shows that under general indifferences, there does not exist a Pareto efficient and group strategy-proof mechanism for Shapley-Scarf markets. Second, even when the core of a market is non-empty, TTC with fixed tie-breaking may not select a core allocation.

Objective indifferences adds structure to the case of general indifferences by constraining any indifferences to be universal among agents. We show that objective indifferences characterizes maximal domains on which TTC with fixed tie-breaking is Pareto efficient, group strategy-proof, and core selecting. This statement contains three distinct results. First, TTC with fixed tie-breaking recovers these three properties on the restriction of objective indifferences. Second, any more general setting causes TTC to lose Pareto efficiency and core selection. Third, any setting in which TTC maintains either of these properties is a subset of an objective indifferences domain. The analogous results hold for group strategy-proofness among preference domains that include both directions of any strict ranking. We will argue that this is not an onerous restriction in market design applications. The last two results are remarkable -- objective indifferences and more restrictive settings are all of the settings on which TTC with fixed tie-breaking is Pareto efficient, core selecting, and group strategy-proof.  In other words, it is subjective indifferences that may differ between agents, and not indifferences in general, which cause TTC with fixed tie-breaking to lose these properties.

While our findings are not necessarily prescriptive nor proscriptive, they may help policymakers decide whether TTC is a good solution for their setting. Consider school choice in San Francisco, which uses a lottery system to assign school seats at most public schools. While current details are not readily available, Abdulkadiroğlu, Featherstone, Niederle, Pathak, and Roth designed a system using TTC.\footnote{See the blog post by Al Roth: https://marketdesigner.blogspot.com/2010/09/san-francisco-school-choice-goes-in.html. As he notes, the team were not privy to the implementation or resulting data.} At some schools, there are seats dedicated to language immersion programs and other seats that are intended for general education. For instance, at West Portal Elementary School, there are roughly 120 seats, approximately 25\% of which are for Cantonese immersion.\footnote{https://web.archive.org/web/20250422170224/https://www.sfchronicle.com/bayarea/article/sfusd-competitive-public-schools-20252957.php} Families separately rank seats by $school\times language$. If all families' preferences can be described by objective indifferences, then TTC maintains its most important properties.

However, the real situation may be more complicated. Some families may be indifferent between bilingual and regular seats, while other families have strict preferences. For example, a family whose children are already bilingual in Cantonese may be indifferent between the two types of seats at West Portal, while another family may have a strict preference for cultural community through the Cantonese bilingual program.\footnote{The relative non-competitiveness of language immersion seats in the SFUSD appears to be a source of anxiety for parents. Seats are further divided into already-bilingual and not already bilingual. While West Portal receives 7.7 requests per open seat, its Cantonese immersion program for already bilingual students receives 7.2 requests per open seat, and 13.3 requests per open seat for students not already bilingual. That is, parents who simply want a seat at West Portal may apply for and receive already-bilingual seats, preventing parents with true desires for these seats from receiving them. Perhaps noticing this, parents have been agitating for expansion of language immersion programs, particularly for Mandarin, although we can only speculate on private motives.} If so, then TTC with fixed tie-breaking is no longer Pareto efficient, core-selecting, nor group strategy-proof. Indeed, it appears San Francisco is ready to change its school choice mechanism away from TTC.\footnote{https://web.archive.org/web/20250524170000/https://www.sfusd.edu/schools/enroll/student-assignment-policy/student-assignment-changes} While there are certainly a variety of reasons San Francisco finds its current mechanism unsatisfactory,\footnote{\citet{Pathak17} writes ``The difficulty in explaining TTC compared to DA is also a reason that Recovery School District in New Orleans switched mechanisms after one year with DA.''} our results can partially rationalize events.

To our knowledge, the first use of the term ``objective indifferences'' was due to \citet{BM99}, who (almost tangentially) note that their probabilistic serial mechanism for random object assignment can accommodate this case. \citet{FSW03} and \citet{CS10} deal with competitive equilibrium in Shapley-Scarf markets with objective indifferences. Objective indifferences is similar to models where objects have capacities, as is sometimes applied in school choice; for example, see \citet{Morrill15} and \citet{AS03}. However, there are important differences which we illustrate in \ref{schoolchoiceappendix}. Object capacities give rise to a priority structure rather than an endowment structure, which changes the way TTC operates. More importantly, our emphasis here is on the objective indifferences setting as a domain of preferences, as will be clear in the results.

Others have have studied Shapley-Scarf markets with indifferences. \citet{Ehlers14} gives a characterization of TTC with fixed tie-breaking under general indifferences. \citet{QW04} are the first to provide an algorithm finding the strict core in the presence of indifferences. \citet{AM11} and \citet{JM12} provide strategy-proof, Pareto efficient, and core selecting families of trading cycle mechanisms. \citet{AK12} and \citet{Plaxton13} propose further generalizations. Fundamentally, the challenge for any mechanism in a Shapley-Scarf market with indifferences is determining which trading cycles to execute from among the many potential trading cycles that indifferences may induce. Using a fixed tie-breaking rule is intuitive and easy to implement, but as \citet{Ehlers14} shows, it comes at the expense of certain desirable properties, including Pareto efficiency and core selection. Though we too study Shapley-Scarf markets with indifferences, we place additional structure on the agents' indifferences and demonstrate how this resolves many of the challenges that indifferences pose.

Our paper makes important contributions to the research in Shapley-Scarf markets along two lines. First, we contribute to the understanding of TTC. While there already exist TTC-like mechanisms that deal with indifferences, TTC is more intuitive and more widely known. Our paper captures exactly when its most important properties are maintained -- under objective indifferences. We thus rationalize its use in many settings where the Shapley-Scarf model is actually applicable, such as in housing or school assignment. Second, we contribute to the understanding and awareness of the objective indifferences setting. While many have implicitly incorporated this model (e.g., objects with capacities), few have explicitly considered objective indifferences as a domain of preferences. Indeed, literature in mechanism design often considers only strict preferences or the full domain of indifferences. We show that objective indifferences adds an interesting and important intermediate case. It realistically models many object allocation problems and has the potential to preserve important properties relative to general indifferences.

Section \ref{sec:model} presents the formal notation. Section \ref{sec:TTC} explains TTC with fixed tie-breaking. Section \ref{sec:Results} provides the main results. Section \ref{sec:concl} concludes. Proofs are in \ref{proofsappendix}.

\section{Model}
\label{sec:model}

In this section, we present the model primitives. First we recount the classical \citet{SS74} setting. Afterwards we introduce the ``objective indifferences'' domain.

We first present the classical Shapley-Scarf model. Let $N=\{1,\ldots,n\}$ be a finite set of agents, with generic member $i$. Let $H=\{h_{1},\ldots,h_{n}\}$ be a set of houses, with generic member $h$. Every agent is endowed with one object, given by a bijection $w:N\rightarrow H$.  An allocation is an assignment of an object to each agent, also a bijection $x:N\rightarrow H$. The set of all allocations given $(N,H)$ is $X(N,H)$, or simply $X$. For any $i\in N$, we use $w_{i}$ as shorthand for $w(i)$ and $x_{i}$ for $x(i)$. Similarly, for any $Q\subseteq N$ we use $w_{Q}$ for $w(Q)=\{w_{i}:i\in Q\}$ and $x_{Q}$ for $x(Q)=\{x_{i}:i\in Q\}$.

Each agent $i$ has some preference ordering over $H$, denoted $R_{i}$. That is, $R_{i}$ is a transitive, complete, reflexive binary relation. As usual, we let $P_{i}$ denote the strict relation and $I_{i}$ denote the indifference relation associated with $R_{i}$. A set of allowable preference orderings in a problem is denoted $\mathcal{R}$, which we call a \textbf{domain}. A preference profile is $R=(R_{1},...,R_{n})\in\mathcal{R}^{n}$. For any $Q\subseteq N$ and preference profile $R$, we denote $R_{Q}=(R_{q})_{q\in Q}$. In this paper we restrict attention to settings where allowable preference profiles are drawn from some $\mathcal{R}^{n}$. That is, every agent has the same set of possible preference relations.

If $\mathcal{R}$ is the set of strict preference relations over $H$, then this is the \textbf{strict preferences domain}. To demonstrate the notation, the strict preference domain can be denoted
\begin{align*}
\mathcal{R}_{\text{strict}}=\left\{ R_{i}:hI_{i}h'\iff h=h'\right\}
\end{align*}
If $\mathcal{R}$ is the set of weak preference relations over $H$, it is the \textbf{general indifferences domain}.

Our main domain is \textbf{objective indifferences}. Let $\mathcal{H}=\{H_{1},H_{2},\ldots,H_{K}\}$ be a partition of $H$. An element $H_{k}$ of a partition is a \textbf{block}. Given $H$ and $\mathcal{H}$, let $\eta:H\rightarrow\mathcal{H}$ be the mapping from a house to the partition element containing it; that is, $\eta(h)=H_{k}$ if $h\in H_{k}$. Given $\mathcal{H}$, we denote the objective indifferences domain as
\begin{align*}
\mathcal{R}(\mathcal{H}):=\left\{ R_{i}:hI_{i}h'\iff\eta(h)=\eta(h')\right\}
\end{align*}
We sometimes suppress $(\mathcal{H})$ from the notation when context makes it clear. Note that an agent is indifferent between houses in the same block of $\hcal$ and has strict preferences between houses in different blocks. Since all $i\in N$ draw from the same set of allowable preferences, these indifferences are ``objective''; all agents must be indifferent between houses in the same block. Thus, we sometimes refer to the ``indifference classes'' of the domain with the understanding that every agent shares the same indifference classes. The next example illustrates.
\begin{example}
    Let $H=\{h_{1},h_{2},h_{3}\}$ and $\mathcal{H}=\left\{ \{h_{1},h_{2}\},\{h_{3}\}\right\} $. Then the two possible preference orderings in $\mathcal{R}(\mathcal{H})$ are given by
    
    \begin{figure}[htbp]
    \centering \begin{table}[H]
    \centering
    \begin{tabular}{@{\hskip 6pt}c@{\hskip 6pt}}
        $R_{\alpha}$ \tabularnewline \hline
        $h_{1},h_{2}$ \tabularnewline
        $h_{3}$ \tabularnewline
    \end{tabular}
    \hspace{1cm}
    \begin{tabular}{@{\hskip 6pt}c@{\hskip 6pt}}
        $R_{\beta}$ \tabularnewline \hline
        $h_{3}$ \tabularnewline
        $h_{1},h_{2}$ \tabularnewline
    \end{tabular}
\end{table}
    \end{figure}
    
    \noindent That is, $h_{1}I_{\alpha}h_{2}P_{\alpha}h_{3}$ and $h_{3}P_{\beta}h_{1}I_{\beta}h_{2}$. Given a set of agents $N=\{1,2,3\}$, objective indifferences requires $R_{i}\in\{R_{\alpha},R_{\beta}\}$ for each $i \in N$.
\end{example}

\subsection{Mechanisms}

This subsection recounts formalities on mechanisms and top trading cycles. Familiar readers may safely skip this subsection.

A \textbf{market} is a tuple $(N,H,w,R)$. A \textbf{mechanism} is a function $f:\mathcal{R}^{n}\rightarrow X$; given a preference profile, it produces an allocation. When it is unimportant or clear from context, we suppress inputs from the notation. For any $i\in N$, let $f_{i}(R)$ denote $i$'s allocated house under $f(R)$. Similarly, for any $Q\subseteq N$, let $f_{Q}(R)=\{f_{i}(R):i\in Q\}$.

Fix a mechanism $f$, a tuple $(N,H,w)$, and a preference domain $\rcal$. We work with the following axioms.

A mechanism is Pareto efficient if it always selects Pareto efficient allocations.

\smallskip{}
    \noindent \textbf{Pareto efficiency (PE)}. For all $R\in\mathcal{R}^{n}$, there
    is no other allocation $x\in X$ such that $x_{i}R_{i}f_{i}(R)$ for all $i\in N$ and $x_{i}P_{i}f_{i}(R)$ for at least one $i\in N$.
    \smallskip{}

Group strategy-proofness requires that no coalition of agents can collectively improve their outcomes by submitting false preferences. Note that in the following definition, we require both the true preferences and misreported preferences to come from the preference domain $\mathcal{R}$.\footnote{This differs from other studies, e.g., \citet{Ehlers02}.}

\smallskip{}
    \noindent \textbf{Group strategy-proofness (GSP)}. For all $R\in\mathcal{R}^{n}$,
    there do not exist $Q\subseteq N$ and $R'=(R'_{Q},R_{-Q})\in\rcal^{n}$ such that $f_{i}(R')R_{i}f_{i}(R)$ for all $i\in Q$ with $f_{i}(R')P_{i}f_{i}(R)$ for at least one $i\in Q$.
    \smallskip{}

Individual rationality models the constraint of voluntary participation. It requires that agents receive a house they weakly prefer to their endowment.

\smallskip{}
    \noindent\textbf{Individual rationality (IR)}. For all $R\in\mathcal{R}^{n}$,
    $f_{i}(R)R_{i}w_{i}$ for all $i\in N$.
\smallskip{}

We also define the core of a market: an allocation is in the core if there is no subset of agents who could benefit from trading their endowments among themselves.
\begin{defn}
    \label{defn:core} An allocation $x$ is \textbf{blocked} if there exists a coalition $Q\subseteq N$ and allocation $y$ such that $y_{Q}=w_{Q}$ and $y_{i}R_{i}x_{i}$ for all $i\in Q$, with $y_{i}P_{i}x_{i}$ for at least one $i\in Q$. An allocation $x$ is in the \textbf{core} of the market if it is not blocked.
    \end{defn}
The weak core requires that all members of a potential coalition are strictly better off.
\begin{defn}
    \label{defn:weakcore} An allocation $x$ is \textbf{weakly blocked} if there exists a coalition $Q\subseteq N$ and allocation $y$ such that $y_{Q}=w_{Q}$ and $y_{i}P_{i}x_{i}$ for all $i\in Q$. An allocation $x$ is in the \textbf{weak core} if it is not weakly blocked.
\end{defn}
\noindent Of course, the weak core also contains the core. The (weak) core property models the restriction imposed by property rights. For example, in housing allocation, existing tenants might have the right to stay in their status quo allocations. Notice that individual rationality excludes blocking coalitions of size 1. The last axiom is core-selecting.

\smallskip{}
    \noindent \textbf{Core-selecting (CS).} For all $R\in\mathcal{R}$, if the
    core of the market is non-empty then $f(R)$ is in the core.
\smallskip{}

\subsection{Maximality}

In Section \ref{sec:Results}, we characterize maximal domains on which TTC with fixed tie-breaking satisfies the axioms. By a ``maximal'' domain, we mean the following.
\begin{defn}
A domain $\rcal$ is \textbf{maximal} for Pareto efficiency and a class of mechanisms $F$ if
\begin{enumerate}
\item every $f\in F$ is Pareto efficient on $\mathcal{R}$, and \item for any $\tilde{\rcal}\supset\rcal$, there is some $f\in F$ that is \emph{not} Pareto efficient on $\tilde{\rcal}$.
\end{enumerate}
\end{defn}
\noindent The definition for any axiom besides Pareto efficiency is analogous. Notice that this definition of maximality depends on both the axiom and the class of mechanisms, which differs from elsewhere in the literature.\footnote{Commonly in the literature, a domain is maximal for Pareto efficiency if there exists \emph{some} mechanism which is Pareto efficient on it, and none on any larger domain.} We define maximality for a class of mechanisms, since our claims deal with TTC for \emph{all} tie-breaking rules. Again note that we only consider preference profiles drawn from $\rcal^{n}$, which is common.\footnote{This excludes settings where different agents may be able to report different sets of preferences.}


\section{Top trading cycles with fixed tie-breaking}
\label{sec:TTC}

In this paper, we analyze top trading cycles (TTC) with fixed tie-breaking on the domains defined in the previous section. For an extensive history of TTC, we refer the reader to \citet{MR24}. We briefly define TTC and TTC with fixed tie-breaking.
\begin{lyxalgorithm}
    \label{alg:ttc} \textbf{Top Trading Cycles}. Consider a market $(N,H,w,R)$ under strict preferences. Draw a graph with $N$ as nodes.
    \begin{enumerate}
    \item Draw an arrow from each agent $i$ to the owner (endowee) of his favorite remaining object. \item There must exist at least one cycle; select one of them. For each agent in this cycle, give him the object owned by the agent he is pointing at. Remove these agents from the graph. \item If there are remaining agents, repeat from step 1.
    \end{enumerate}
\end{lyxalgorithm}
\noindent We denote the resulting allocation as $\text{TTC}(R)$.

TTC is well-defined only with strict preferences, as Step 1 requires a unique favorite object. In practice, a \textbf{fixed tie-breaking profile} $\succ$ is often proposed to resolve indifferences. Given $N$, let $\succ=(\succ_{1},\ldots,\succ_{n})$, where each $\succ_{i}$ is a strict linear order over $N$. This linear order will be used to break indifferences between objects (based on their owners). For any preference relation $R_{i}$ and tie-breaking rule $\succ_{i}$, let $R_{i,\succ}$ be given by the following. For any $j\not=j'$, let $w_{j}R_{i,\succ}w_{j'}$ if
\begin{enumerate}
    \item $w_{j}P_{i}w_{j'}$, or
    \item $w_{j}I_{i}w_{j'}$ and $j\succ_{i}j'$.
\end{enumerate}
For all $j$, let $w_{j}R_{i,\succ}w_{j}$.

Then $R_{i,\succ}$ is a strict linear order over the individual houses. Note that $R_{i,\succ}$ is actually $R_{i,\succ_{i}}$ as it depends only on $i$'s tie-breaking rule, but we suppress this at minimal risk of confusion. Example \ref{ex:tiebreak} illustrates how we combine an agent's preferences and the tie-breaking rule to construct tie-broken preferences.

\begin{example}
\label{ex:tiebreak}
    Let $N=\{1,2,3,4\}$. Agent 1's preferences $R_{1}$ and tie-breaking rule $\succ_{1}$ are shown below. In our lists of preference relations, each line represents an indifference class, and the houses on any line are strictly preferred to houses on lines below them. For example, the representation of $R_{1}$ below indicates $w_{3}I_{1}w_{4}P_{1}w_{1}I_{1}w_{2}$.
    
    \begin{figure}[htbp]
    \centering \noindent\makebox[1\textwidth]{%
\begin{minipage}[t]{0.15\textwidth}%
\centering %
\begin{tabular}{c}
$R_{1}$\\
\hline 
$w_{3},w_{4}$\\
$w_{1},w_{2}$\\
\\
\\
\end{tabular}%
\end{minipage}%
\begin{minipage}[t]{0.05\textwidth}%
 $+$ %
\end{minipage}%
\begin{minipage}[t]{0.15\textwidth}%
\centering %
\begin{tabular}{c}
$\succ_{1}$\\
\hline 
$1$\\
$2$\\
$3$\\
$4$\\
\end{tabular}%
\end{minipage}%

\begin{minipage}[t]{0.05\textwidth}%
 $\rightarrow$ %
\end{minipage}%

\begin{minipage}[t]{0.15\textwidth}%
 \centering %
\begin{tabular}{c}
$R_{1,\succ}$\\
\hline 
$w_{3}$\\
$w_{4}$\\
$w_{1}$\\
$w_{2}$\\
\end{tabular}%
\end{minipage}%
}
\smallskip{} \label{fig:example_making_tb_rankings}
    \end{figure}
    
    \noindent Since $w_{3}I_{1}w_{4}$ and $3\succ_{1}4$, we have $w_{3}P_{1,\succ}w_{4}$.
    Likewise, since $w_{1}I_{1}w_{2}$ and $1\succ_{1}2$, we have $w_{1}P_{1,\succ}w_{2}$. Therefore agent 1's complete tie-broken preferences $R_{1,\succ}$ are given by $w_{3}\;P_{1,\succ}\;w_{4}\;P_{1,\succ}\;w_{1}\;P_{1,\succ}\;w_{2}$.
\end{example}

\smallskip{}

Given a preference profile $R\in\rcal^{n}$ and a tie-breaking profile $\succ$, let $R_{\succ}=(R_{1,\succ},\ldots,R_{n,\succ})$. \textbf{TTC with fixed tie-breaking ($\text{TTC}_{\succ}$)} is $\text{TTC}_{\succ}(R)\equiv\text{TTC}(R_{\succ})$. That is, the tie-breaking profile is used to generate strict preferences, and TTC is applied to the resulting strict preference profile. Formally, each tie-breaking profile $\succ$ generates a different \ttc mechanism, so TTC with fixed tie-breaking is a class of mechanisms, $\left\{ \text{TTC}_{\succ}\right\} _{\succ}$. For a given $R$ and $\succ$, we use $\text{TTC}_{\succ}(R)$ to refer both to the step-by-step procedure of \ttc and to the final allocation it generates. The following example illustrates how \ttc works in the objective indifferences domain.
\begin{example}
\label{ttcwithtb}

Let $N=\{1,2,3,4\}$ and $\mathcal{H}=\{ \{w_1\}, \{w_2, w_3\}, \{w_4\} \}$. The preference profile $R$ and tie-breaking profile $\succ=(\succ_{1},\succ_{2},\succ_{3},\succ_{4})$ are shown below. $R$ and $\succ$ are combined in the manner shown in Example \ref{ex:tiebreak} to construct the tie-broken preference profile $R_{\succ}$. Recall that $\text{TTC}_{\succ}(R)$ is equivalent to $\text{TTC}(R_{\succ})$.

\begin{figure}[H]
\centering

\begin{subfigure}[t]{0.3\textwidth}
    \centering
    \centering
\setlength{\tabcolsep}{10pt}
\begin{tabular}{@{\hskip 6pt}c@{\hskip 6pt}c@{\hskip 6pt}c@{\hskip 6pt}c@{\hskip 6pt}}
    $R_1$ & $R_2$ & $R_3$ & $R_4$ \\
    \hline
    $w_2,w_3$ & $w_1$ & $w_1$ & $w_2,w_3$ \\
    $w_1$     & $w_2,w_3$ & $w_4$ & $w_4$ \\
    $w_4$     & $w_4$ & $w_2,w_3$ & $w_1$ \\
    &&&
\end{tabular}

    \caption*{\small Preference profile $R$}
\end{subfigure}
\hfill
\begin{subfigure}[t]{0.3\textwidth}
    \centering
    \centering
\setlength{\tabcolsep}{10pt}
\begin{tabular}{@{\hskip 6pt}c@{\hskip 6pt}c@{\hskip 6pt}c@{\hskip 6pt}c@{\hskip 6pt}}
    $\succ_1$ & $\succ_2$ & $\succ_3$ & $\succ_4$ \\
    \hline
    $2$ & $1$ & $3$ & $3$ \\
    $1$ & $2$ & $2$ & $1$ \\
    $3$ & $3$ & $1$ & $2$ \\
    $4$ & $4$ & $4$ & $3$
\end{tabular}

    \caption*{\small Tie-breaking profile $\succ$}
\end{subfigure}
\hfill
\begin{subfigure}[t]{0.3\textwidth}
    \centering
    \centering
\setlength{\tabcolsep}{10pt}
\begin{tabular}{@{\hskip 6pt}c@{\hskip 6pt}c@{\hskip 6pt}c@{\hskip 6pt}c@{\hskip 6pt}}
    $R_{1,\succ}$ & $R_{2,\succ}$ & $R_{3,\succ}$ & $R_{4,\succ}$ \\
    \hline
    $w_2$ & $w_1$ & $w_1$ & $w_3$ \\
    $w_3$ & $w_2$ & $w_4$ & $w_2$ \\
    $w_1$ & $w_3$ & $w_3$ & $w_4$ \\
    $w_4$ & $w_4$ & $w_2$ & $w_1$
\end{tabular}
    \caption*{\small Tie-broken preference profile $R_\succ$}
\end{subfigure}
\end{figure}

\vspace{1em}

\noindent{}%
\begin{minipage}[c]{0.25\textwidth}%
\centering \begin{tikzpicture}[
    > = Stealth, 
    shorten > = 1pt, 
    auto,
    node distance = 1.3cm, 
    semithick 
]
    \node[circle, draw, inner sep=2pt] (1) {$1$};
    \node[circle, draw, inner sep=2pt, right of=1] (2) {$2$};
    \node[circle, draw, inner sep=2pt, below of=1] (3) {$3$};
    \node[circle, draw, inner sep=2pt, below of=2] (4) {$4$};
    
    \path[->] (1) edge [bend left] (2);
    \path[->] (1) edge [dashed, color = red, bend right] (3);
    \path[->] (2) edge [bend left] (1);
    \path[->] (3) edge [bend right] (1);
    \path[->] (4) edge [dashed, color = red, bend right] (2);
    \path[->] (4) edge [bend left] (3);
\end{tikzpicture} %
\end{minipage}\hfill{}%
\begin{minipage}[c]{0.7\textwidth}%
\vspace*{0pt} 
\underline{Step 1 of $\text{TTC}_{\succ}(R)$:} \\
Each agent points to the owner of their favorite house according
to their \emph{tie-broken} preferences, represented by black arrows.
Red dashed arrows represent indifferences between $w_{2}$ and $w_{3}$.
Agents 1 and 2 form a cycle and therefore swap houses. %
\end{minipage}

\vspace{3em}

\noindent{}%
\begin{minipage}[c]{0.25\textwidth}%
\centering \begin{tikzpicture}[
    > = Stealth, 
    shorten > = 1pt, 
    auto,
    node distance = 1.3cm, 
    semithick 
]
    \node[circle, draw, inner sep=2pt] (3) {$3$};
    \node[circle, draw, inner sep=2pt, right of=3] (4) {$4$};
    
    \path[->] (3) edge [bend left] (4);
    \path[->] (4) edge [bend left] (3);
\end{tikzpicture} %
\end{minipage}\hfill{}%
\begin{minipage}[c]{0.7\textwidth}%
\vspace*{0pt}
\underline{Step 2 of $\text{TTC}_{\succ}(R)$:} \\ After removing the agents assigned in Step 1, the remaining agents (3 and 4) point to the owner of their favorite remaining house. They form a cycle, and therefore swap houses. Since every agent has been assigned to a house, the \ttc procedure ends. The resulting allocation is $x=(w_{2},w_{1},w_{4},w_{3})$. %
\end{minipage}
\end{example}


\section{Results}\label{sec:Results}

With general indifferences, $\text{TTC}_{\succ}$ mechanisms are not Pareto efficient, core-selecting, nor group strategy-proof. We provide simple examples below to illustrate these failures. However, we show that in the objective indifferences domain, \ttc mechanisms satisfy all three properties. We also show that expanding beyond this domain results in the loss of each property. Strikingly, any domain satisfying any of these properties must be a subset of objective indifferences. That is, objective indifferences characterizes the set of maximal domains on which \ttc mechanisms are PE and CS and characterizes the set of ``symmetric-maximal'' (defined below) domains on which \ttc mechanisms are GSP.

\subsection{Pareto efficiency and core-selecting}

When we relax the assumption of strict preferences and allow for general indifferences, \ttc loses two of its most appealing properties: Pareto efficiency and core-selecting. However, in the intermediate case of objective indifferences, \ttc retains these two properties. Moreover, on \emph{any} larger domain, $\text{TTC}_{\succ}$ loses both Pareto efficiency and core-selecting. Thus, we show that it is not indifferences per se, but rather \textit{subjective} evaluations of indifferences which cause \ttc to lose these properties.

We first demonstrate that $\text{TTC}_{\succ}$ mechanisms are not Pareto efficient under general indifferences. Example \ref{ex:indiff-pe} gives the simplest case.

\begin{example}
\label{ex:indiff-pe} Let $N=\{1,2\}$. The preference profile $R=(R_{1},R_{2})$, tie-breaking profile $\succ=(\succ_{1},\succ_{2})$, and tie-broken preference profile $R_{\succ}=(R_{1,\succ},R_{2,\succ})$ are shown below.

\begin{figure}[H]
\centering

\begin{subfigure}[t]{0.3\textwidth}
    \centering
    \centering
\setlength{\tabcolsep}{6pt}
\begin{tabular}{cc}
    $R_{1}$  & $R_{2}$ \\
    \hline 
    $w_{1},w_{2}$  & $w_{1}$ \\
                   & $w_{2}$ 
\end{tabular}
    \caption*{\small Preference profile $R$}
\end{subfigure}
\hfill
\begin{subfigure}[t]{0.3\textwidth}
    \centering
    \centering
\setlength{\tabcolsep}{6pt}
\begin{tabular}{c}
    $\succ_{1} = \succ_{2}$ \\
    \hline 
    $1$ \\
    $2$ 
\end{tabular}

    \caption*{\small Tie-breaking profile $\succ$}
\end{subfigure}
\hfill
\begin{subfigure}[t]{0.3\textwidth}
    \centering
    \centering
\setlength{\tabcolsep}{6pt}
\begin{tabular}{cc}
    $R_{1,\succ}$  & $R_{2,\succ}$ \\
    \hline 
    $w_{1}$  & $w_{1}$ \\
    $w_{2}$  & $w_{2}$ 
\end{tabular}
    \caption*{\small Tie-broken preference profile $R_\succ$}
\end{subfigure}
\end{figure}

\noindent The \ttc allocation is $x=(w_{1},w_{2})$, which is Pareto
dominated by $y=(w_{2},w_{1})$.
\end{example}

This example demonstrates the underlying reason that \ttc fails PE under general indifferences: tie-breaking rules may not take advantage of Pareto gains made possible by the agents' indifferences. However, under objective indifferences, if any agent is indifferent between two houses, then all agents are indifferent between those two houses. Consequently, objective indifferences rules out situations like the one shown in Example \ref{ex:indiff-pe}.

Under general indifferences, the set of core allocations may not be a singleton; there may be no core allocations or there may be multiple. As Example \ref{ex:indiff-pe} demonstrates, even when the core of the market is non-empty, \ttc may still fail to select a core allocation.\footnote{It is straightforward to see that $y=(w_{2},w_{1})$ is in the core of the market and $x=(w_{1},w_{2})$ is not.} However, under objective indifferences, if the core of a market is non-empty, then $\text{TTC}_{\succ}$ always selects a core allocation.

In fact, the objective indifferences setting characterizes the entire set of maximal domains on which \ttc mechanisms are Pareto efficient or core-selecting. That is, if all \ttc mechanisms are PE or CS on a domain $\rcal$, then $\rcal$ must be a subset of some objective indifferences domain. Conversely, for any superset of an objective indifferences domain, there is some \ttc mechanism (that is, some $\succ$) that is not PE or CS.
\begin{thm}
\label{thm:pecs} The following are equivalent:
\begin{enumerate}
\item $\rcal$ is an objective indifferences domain. \item $\rcal$ is a maximal domain on which \ttc mechanisms are Pareto efficient. \item $\rcal$ is a maximal domain on which \ttc mechanisms are core-selecting.
\end{enumerate}
\end{thm}
\begin{proof}
\ref{paretocoreappendix}.
\end{proof}
The full proof is in the appendix, but the intuition is simple. The objective indifferences domain precludes possibilities such as Example \ref{ex:indiff-pe}, and any larger domain inevitably introduces the
possibility of such a pair.

It follows from \citet{Sonmez99} that under objective indifferences, the core of a market is \textit{essentially single-valued} when it exists. That is, for any two allocations $x$ and $y$ in the core of a market, we have $x_{i}I_{i}y_{i}$ for all agents $i$. In our proof of Theorem \ref{thm:pecs}, we also prove this claim directly. Since the core is essentially single-valued, under objective indifferences the core can be interpreted as a unique mapping from agents to house types. In other words, the core allocations are permutations of one another, where in each core allocation an agent always receives a house from the same indifference class.
\begin{cor}
\label{cor:uniquecore} If $x$ and $y$ are in the core of an objective indifferences market, then $x_{i}I_{i}y_{i}$ for all $i\in N$.
\end{cor}
\begin{proof}
\ref{paretocoreappendix}.
\end{proof}
Though all \ttc mechanisms are core-selecting under objective indifferences, the core of the market may still be empty, as the following simple example shows.
\begin{example}
\label{ex:emptycore} Let $N=\{1,2,3\}$. It is easy to verify that for the preference profile $R=(R_{1},R_{2},R_{3})$ shown below, there are no core allocations.

\begin{figure}[h!tbp]
\centering \begin{tabular}{ccc}
    $R_{1}$  & $R_{2}$ & $R_{3}$\\
    \hline 
    $w_{2}, w_{3}$  & $w_{1}$ & $w_{1}$\\
    $w_{1}$  & $w_{2}, w_{3}$ &  $w_{2}, w_{3}$
\end{tabular}

\end{figure}

\noindent Any allocation $x$ such that $x_{1}\in\{w_{1},w_{2}\}$
is blocked by $Q=\{1,3\}$ and $y_Q=(w_{3},w_{1})$. Similarly, any allocation $x$ such that $x_{1}=w_{3}$ is blocked by $Q=\{1,2\}$ and $y_Q=(w_{2},w_{1})$.
\end{example}
Even when the core of an objective indifferences market is empty, all \ttc mechanisms select an allocation in the weak core of the market. This is a known fact even under general indifferences; we reproduce it here and provide a proof for convenience.
\begin{prop}
\label{prop:weakcore} For any market, the weak core is non-empty and \ttc mechanisms select an allocation in the weak core.
\end{prop}
\begin{proof}
\ref{paretocoreappendix}.
\end{proof}

\subsection{Group strategy-proofness}

\ttc also loses group strategy-proofness once we move from strict
preferences to weak preferences. However, in the intermediate case of objective indifferences, \ttc recovers group strategy-proofness. Further, \ttc mechanisms are not GSP in any larger ``symmetric'' domain. We say that a domain $\rcal$ is symmetric if, when $h_{1}P_{i}h_{2}$ is possible, then so is $h_{2}P_{i}h_{1}$. We will informally argue that this is not an onerous modeling restriction.

First we present a simple example demonstrating that under general indifferences, \ttc mechanisms are not group strategy-proof. Notice that in Example \ref{ex:indiff-pe}, agent 1 can misreport $w_{2}R_{1}'w_{1}$ to benefit agent 2 without harming himself. As a less trivial example, we also provide Example \ref{ex:indiff-gsp}.
\begin{example}
\label{ex:indiff-gsp}

Let $N=\{1,2,3\}$ and let $Q=\{1,3\}$. For the preference profile $R=(R_{1},R_{2},R_{3})$ and tie-breaking profile $\succ=(\succ_{1},\succ_{2},\succ_{3})$ shown below, the \ttc allocation is $x=(w_{2},w_{1},w_{3})$. However, if agent 1 were to report $R'_{1}$, then for $R'=(R'_{1},R_{2},R_{3})$ the \ttc allocation is $y=(w_{3},w_{2},w_{1})$. Note that $y_{3}P_{3}x_{3}$ and $y_{1}I_{1}x_{1}$, so \ttc is not GSP.

\begin{figure}[H]
\centering

\begin{subfigure}[t]{0.3\textwidth}
    \centering
    \centering
\setlength{\tabcolsep}{6pt}
\begin{tabular}{ccc}
    $R_{1}$  & $R_{2}$  & $R_{3}$  \\
    \hline 
    $w_{2},w_{3}$  & $w_{1}$  & $w_{1}$  \\
    $w_{1}$        & $w_{2}$  & $w_{2}$  \\
                   & $w_{3}$  & $w_{3}$  
\end{tabular}
    \caption*{\small Preference profile $R$}
\end{subfigure}
\hfill
\begin{subfigure}[t]{0.3\textwidth}
    \centering
    \centering
\setlength{\tabcolsep}{6pt}
\begin{tabular}{ccc}
    $R'_{1}$  & $R_{2}$ & $R_{3}$ \\
    \hline 
    $w_{3}$  & $w_{1}$ & $w_{1}$ \\
    $w_{2}$  & $w_{2}$ & $w_{2}$ \\
    $w_{1}$  & $w_{3}$ & $w_{3}$
\end{tabular}
    \caption*{\small Preference profile $R'$}
\end{subfigure}
\hfill
\begin{subfigure}[t]{0.3\textwidth}
    \centering
    \centering
\setlength{\tabcolsep}{6pt}
\begin{tabular}{c}
    $\succ_{1} = \succ_{2} = \succ_{3}$ \\
    \hline 
    $1$ \\
    $2$ \\
    $3$
\end{tabular}
    \caption*{\small Tie-breaking profile $\succ$}
\end{subfigure}
\end{figure}

\end{example}

Objective indifferences excludes situations like Example \ref{ex:indiff-gsp} in two ways. First, it eliminates the possibility that one agent is indifferent between two houses while another agent has a strict preference. Second, it constrains the possible set of misreports available to a manipulating coalition, since agents can \textit{only} report indifference among all houses of the same type. Our next result characterizes the set of symmetric-maximal domains on which \ttc mechanisms are GSP.

Before presenting our result, we define ``symmetric'' and ``symmetric-maximal'' domains.
\begin{defn}
A domain $\rcal$ is \textbf{symmetric} if for any $h_{1},h_{2}\in H$, if there exists $R_{i}\in\rcal$ such that $h_{1}P_{i}h_{2}$, then there also exists $R_{i}'\in\rcal$ such that $h_{2}P_{i}'h_{1}$.
\end{defn}
\begin{defn}
A domain $\rcal$ is \textbf{symmetric-maximal} for Pareto efficiency and a class of mechanisms $F$ if
\begin{enumerate}
\item $\rcal$ is symmetric,
\item every $f \in F$ is Pareto efficient on $\rcal$, and 
\item for any symmetric $\tilde{\rcal}\supset\rcal$, there is some $f\in F$ that is \emph{not} Pareto efficient on $\tilde{\rcal}$.
\end{enumerate}
The definition for any axiom besides Pareto efficiency is again analogous.
\end{defn}
In most practical applications, symmetry is a natural restriction to place on the domain. If it is possible that agents might report strictly preferring some house $h$ to another house $h'$, we should not preclude the possibility they strictly prefer $h'$ to $h$. Indeed, a central principle of mechanism design is that preferences are unknown and must be elicited. It is easy to see that objective indifferences domains are symmetric. Compared to maximality, symmetric-maximality restricts the possible expansions of objective indifferences domains that we must consider. The symmetry restriction is relatively weak; it does not enforce strict preferences nor general indifferences, as the next example illustrates.

\begin{example}
Let $H=\{h_{1},h_{2},h_{3}\}$ and $\mathcal{H}=\left\{ \{h_{1}\},\{h_{2},h_{3}\}\right\} $. Consider the objective indifferences domain $\mathcal{R}(\mathcal{H})=\left\{ h_{1}Ph_{2}Ih_{3},h_{2}Ih_{3}Ph_{1}\right\} $. Let $R_{\alpha}=h_{1}Ph_{2}Ph_{3}$ and $R_{\beta}=h_{1}Ph_{3}Ph_{2}$. Then $\mathcal{R}(\mathcal{H})\cup\{R_{\alpha}\}$ is not a symmetric domain, since it allows $h_{2}Ph_{3}$ but not $h_{3} P h_{2}$. However, $\mathcal{R}(\mathcal{H})\cup\{R_{\alpha},R_{\beta}\}$ \emph{is} a symmetric domain. Further, notice that $\mathcal{R}(\mathcal{H})\cup\{R_{\alpha},R_{\beta}\}$ is not general indifferences nor strict preferences. Nor does it belong to some other objective indifferences domain, as it allows agents to report both $h_{2}Ih_{3}$ and $h_{2}Ph_{3}$.
\end{example}
Our second main result is that objective indifferences domains are the only symmetric-maximal domains on which every \ttc mechanism is group strategy-proof.
\begin{thm}
\label{thm:gsp} $\rcal$ is a symmetric-maximal domain on which \ttc mechanisms are group strategy-proof if and only if $\rcal$ is an objective indifferences domain.
\end{thm}
\begin{proof}
\ref{gspappendix}.
\end{proof}
Our proof uses similar reasoning to the proof that TTC is group strategy-proof under strict preferences contained in \citet{ST24}. Any coalition requires a ``first mover'' to misreport, but this agent must receive an inferior house to the one he originally received. Under objective indifferences, this ``inferior'' house must actually be inferior according to his true preferences. Under more general indifferences, the situation in Example \ref{ex:indiff-pe} again inevitably arises.

In the following example, we note that objective indifferences domains are \emph{not} maximal domains on which $\text{TTC}_{\succ}$ mechanisms are GSP.
\begin{example}
    Let $N=\{1,2\}$, $H=\{h_{1},h_{2}\}$, and $\mathcal{H}=\{\{h_{1},h_{2}\}\}$. Suppose $\rcal'=\rcal(\mathcal{H})\cup\{h_{1}Ph_{2}\}=\left\{ h_{1}Ih_{2},h_{1}Ph_{2}\right\} $. That is, expand the objective indifferences domain induced by $\hcal$ to include the ordering $h_{1}Ph_{2}$. Note that this expanded domain is not symmetric, since $\rcal'$ does not also contain the preference ordering $h_{2}Ph_{1}$.
    
    Let $\succ=((1\succ_{1}2),(1\succ_{2}2))$. We will show that for any market $(N,H,w,R)$, \ttc is group strategy-proof. It is straightforward to show that the same is true for the remaining 3 possible tie-breaking profiles.
    
    Without loss of generality, assume $w_{i}=h_{i}$. If both agents have the same preferences, then there is clearly no profitable group manipulation. Consider the following two possible preference profiles:
    
    \begin{figure}[htbp]
    \centering \noindent
\makebox[\textwidth][c]{%
  \begin{minipage}[t]{0.35\textwidth}
    \centering
    \begin{tabular}{cc}
      $R_{1}$ & $R_{2}$ \\
      \hline 
      $h_{1}$ & $h_{1},h_{2}$ \\
      $h_{2}$ & 
    \end{tabular}
  \end{minipage}
  \hspace{2em}
  \begin{minipage}[t]{0.05\textwidth}
    \centering
    \text{or}
  \end{minipage}
  \hspace{2em}
  \begin{minipage}[t]{0.35\textwidth}
    \centering
    \begin{tabular}{cc}
      $R_{1}$ & $R_{2}$ \\
      \hline 
      $h_{1},h_{2}$ & $h_{1}$ \\
                   & $h_{2}$
    \end{tabular}
  \end{minipage}
}
    \end{figure}
    
    In the first case, the \ttc allocation is $x=(h_{1},h_{2})$, so both agents receive one of their most preferred houses. Therefore, it is not possible for either agent to strictly improve. In the second case, the \ttc allocation is $x=(h_{1},h_{2})$. It would benefit agent 2 for agent 1 to rank $h_{2}$ above $h_{1}$, since agent 1 is indifferent between $h_{1}$ and $h_{2}$. However, this is not possible since $(h_{2}Ph_{1})\notin\rcal'$.
\end{example}

Even though objective indifferences domains are not maximal for GSP, we have already argued that the symmetric restriction is not onerous. In most mechanism design settings, it is strange to allow participants to strictly prefer one object to another, but not vice versa.\footnote{One reasonable exception is contracts where the objects may include payments; e.g., a seat at a school with or without a scholarship.}


\section{Conclusion}

\label{sec:concl}

Our main set of results show that objective indifferences domains are the maximal domains on which \ttc mechanisms are Pareto efficient, core-selecting, and group strategy-proof. While it may not be surprising that objective indifferences preserves these properties, it is quite remarkable that the maximal domains on which \ttc satisfies these three distinct properties essentially coincide.

In markets where one could reasonably assume that any indifferences are shared among all agents, \ttc is a sensible choice of mechanism. Regardless of the tie-breaking rule, \ttc will be Pareto efficient, core-selecting, and group strategy-proof. While a number of other mechanisms generalize TTC and retain various properties in general indifferences, \ttc remains computationally efficient and easier to explain. However, in any more general domain of preferences, TTC inevitably loses some of its appeal. We do not interpret our results as necessarily prescriptive nor proscriptive. We simply characterize exactly when TTC retains its most desirable properties.

Our paper opens several new lines of inquiry. First, we believe that studying matching markets with constrained indifferences is an exciting avenue for future research. In many real-world matching markets, agents have indifferences, but often with a structure imposed by the particular market. Adding structure to general indifferences is not only theoretically interesting, but could improve policy choices. For instance, there may be tradeoffs in the selection of the partition $\mathcal{H}$ given the set of objects $H$. In some cases, there may be some ambiguity: are two dorms with the same floor plan but on different floors of the same building equivalent? Inappropriately combining indifference classes might lead to efficiency losses in the spirit of Example \ref{ex:indiff-pe}. On the other hand, splitting indifference classes might allow group manipulations like in Example \ref{ex:indiff-gsp}. We leave formal results as future work. We also leave an axiomatic characterization of \ttc on objective indifferences domains as future work.

\newpage{}

\newpage{}


\appendix

\global\long\def\thesection{Appendix \Alph{section}}%
\global\long\def\thesubsection{Appendix \Alph{section}.\arabic{subsection}}%
\global\long\def\thesubsubsection{Appendix \Alph{section}.\arabic{subsection}.\arabic{subsubsection}}%

\section{Proofs }\label{proofsappendix}

We provide proofs for the results in the main text. Note that individual rationality (IR) of $\text{TTC}_{\succ}$ follows immediately from IR of TTC and the fact that $\text{TTC}_{\succ}(R)\equiv \text{TTC}(R_{\succ})$. That is, TTC is IR according to $R_{\succ}$, which implies \ttc is IR according to $R$.

Given a market $(N,H,w,R)$ and mechanism $\text{TTC}_{\succ}$, let $S_{k}(R)$ be the agents in the $k$th cycle executed in the process of $\text{TTC}_{\succ}(R)$. Denote $\bar{S}_{k-1}(R)=\cup_{\ell=1}^{k-1}S_{\ell}(R)$. Where the risk of confusion is minimal (generally when not dealing with strategy-proofness), we suppress the dependence on $R$ and simply refer to $S_{k}$ and $\bar{S}_{k-1}$. While $\text{TTC}_{\succ}(R)$ may execute cycles in different orders, the cycles are always the same for a given market, so the order will be unimportant.

We will appeal to the following fact.
\begin{fact}
\label{fact:assignedbefore} Fix a market $(N,H,w,R)$ and a tie-breaking profile $\succ$. Let $x=\text{TTC}_{\succ}(R)$. If $i\in S_{k}$ and $x_{j}P_{i}x_{i}$, then $j\in\bar{S}_{k-1}$.
\end{fact}
In words, if $i$ strictly prefers some object to $x_{i}$, that object must have been assigned in an earlier cycle. Fact \ref{fact:assignedbefore} follows from the definitions: $x_{j}P_{i}x_{i}$ implies $x_{j}P_{i,\succ}x_{i}$, and $\text{TTC}_{\succ}(R)\equiv\text{TTC}(R_{\succ})$. If a better object than $x_{i}$ is available, $i$ will point at it; thus $i$ cannot be assigned to $x_{i}$ until the better object is gone.

We also make use of the following technical lemma. We do not think it is of interest on its own, but it will be useful in the proceeding proofs.
\begin{lem}
\label{lem:xw} Fix $(N,H,w)$ and a domain $\rcal$. For any preference relations $R_{\alpha},R_{\beta}\in\rcal$ and (not necessarily distinct) houses $h_{1},h_{2}\in H$, suppose $A$ and $B$ satisfy \[ \{i\in N:w_{i}P_{\alpha}h_{1}\}\subseteq A\subseteq\{i\in N:w_{i}R_{\alpha}h_{1}\} \] and \[ \{i\in A^{c}:w_{i}P_{\beta}h_{2}\}\subseteq B\subseteq\{i\in A^{c}:w_{i}R_{\beta}h_{2}\}. \] That is, $A$ includes all agents whose endowments are strictly preferred to $h_{1}$ according to $R_{\alpha}$, and perhaps some whose endowments are indifferent. Among the remainder, $B$ includes all those whose endowments are strictly preferred to $h_{2}$ according to $R_{\beta}$, and perhaps some whose endowments are indifferent.

Let $R$ be any preference profile in $\rcal^n$ such that $R_{i}=R_{\alpha}$ for all $i\in A$ and $R_{i}=R_{\beta}$ for all $i\in B$. Let $\succ$ be any tie-breaking profile such that for all $i\in N$, $i\succ_{i}j$ for all $j\neq i$. That is, every agent's tie-breaking rule prioritizes himself first. Denote $x = \text{TTC}_{\succ}(R)$. Then $x_{i}=w_{i}$ for all $i\in A\cup B$.
\end{lem}
\begin{proof}[Proof of Lemma \ref{lem:xw}]
We first note that since $i\succ_{i}j$ for all $i\neq j$, an immediate consequence of $\text{TTC}_{\succ}$ is that if $x_{i}I_{i}w_{i}$, then $x_{i}=w_{i}$. To see this, observe that if $w_{i}$ is available, $i$ will never point to any other $hI_{i}w_{i}$.

First we show that $x_{i}=w_{i}$ for all $i\in A$. Toward a contradiction, suppose that $U:=\{i\in A:x_{i}\neq w_{i}\}$ (for ``upper'') is non-empty. IR and the fact above imply $U=\{i\in A:x_{i}P_{i}w_{i}\}$. Fix some $i\in U$ such that $w_{i}R_{\alpha}w_{i'}$ for all $i'\in U$; informally, $i$'s endowment is one of $U$'s favorites among their endowments. Since $i\in A$, by construction $R_{i}=R_{\alpha}$; therefore $x_{i}P_{i}w_{i}$ implies $x_{i}P_{\alpha}w_{i}$. Consider $j\in N$ endowed with $x_{i}$; i.e., $x_{i}=w_{j}$. We have $(x_{i}=)w_{j}P_{\alpha}w_{i}$ and $w_{i}P_{\alpha}h_{1}$ by construction, so $j\in A$. We also have $x_{j}\neq w_{j}$, so $j\in U$. But since $w_{j}P_{\alpha}w_{i}$, $j \in U$ contradicts the assumption that $w_{i}R_{\alpha}w_{i'}$ for all $i'\in U$.

Next we show that $x_{i}=w_{i}$ for all $i\in B$. This follows nearly the same logic as above. Suppose that $U=\{i\in B:x_{i}\neq w_{i}\}=\{i\in B:x_{i}P_{i}w_{i}\}$ is non-empty. Fix some $i\in U$ such that $w_{i}R_{\beta}w_{i'}$ for all $i'\in U$. Since $i\in B$, by construction $R_{i}=R_{\beta}$; therefore, $x_{i}P_{i}w_{i}$ implies $x_{i}P_{\beta}w_{i}$. Consider $j\in N$ endowed with $x_{i}$; i.e., $x_{i}=w_{j}$. It must be that $j\in A^{c}$, because we showed that $x_{i}=w_{i}$ for all $i\in A$ and $x_{j}\neq w_{j}$. We have $w_{j}P_{\beta}w_{i}$ and $w_{i}P_{\beta}h_{2}$ by construction, so $j\in B$. We also have $x_{j}\neq w_{j}$, so $j\in U$. But since $w_{j}P_{\beta}w_{i}$, $j \in U$ contradicts the assumption that $w_{i}R_{\beta}w_{i'}$ for all $i'\in U$.
\end{proof}

\subsection{Pareto efficiency and core-selecting }\label{paretocoreappendix}

\textbf{Theorem \ref{thm:pecs}.} \emph{The following are equivalent:
\begin{enumerate}
\item $\rcal$ is an objective indifferences domain. \item $\rcal$ is a maximal domain on which \ttc mechanisms are Pareto efficient. \item $\rcal$ is a maximal domain on which \ttc mechanisms are core-selecting.
\end{enumerate}}
\noindent Fix $(N,H,w)$. The result is trivial for $|N|=1$, so assume
$|N|\geq2$.
\begin{proof}[Proof of 1$\iff$ 2]
First we show that for any objective indifferences domain, all \ttc mechanisms are PE. Fix some tie-breaking rule $\succ$. Let $\mathcal{H}$ be any partition of $H$ and let $R\in\mathcal{R}(\mathcal{H})^{n}$. If $\mathcal{H}=\{H\}$, the result is trivial, so suppose the partition has at least two blocks.

Let $x=\text{TTC}_{\succ}(R)$, and suppose that some feasible allocation $y$ Pareto dominates $x$. Let $U=\{i\in N:y_{i}P_{i}x_{i}\}$ be the set of agents who are strictly better off under $y$, which must be non-empty. Let $k$ be the first step in the process of $\text{TTC}_{\succ}(R)$ that an agent in $U$ is assigned. Formally, $\bar{S}_{k-1}\cap U=\emptyset$ and $S_{k}\cap U\neq\emptyset$. Consider some $j\in S_{k}\cap U$. We have $y_{j}P_{j}x_{j}$, so Fact \ref{fact:assignedbefore} implies that $\{i:x_{i}\in\eta(y_{j})\}\subseteq\bar{S}_{k-1}$. That is, anyone assigned to a house in $\eta(y_{j})$ was assigned before $j$. Because $\eta(y_{j})$ is finite and $\eta(y_{j})\neq\eta(x_{j})$, there must be an agent $\ell\in\bar{S}_{k-1}$ for whom $x_{\ell}\in\eta(y_{j})$ but $y_{\ell}\notin\eta(y_{j})$. In words, in order to assign $j$ to $y_{j}$, someone who originally received an object in $\eta(y_{j})$ under $x$ must receive something else under $y$. Therefore $\neg(y_{\ell}I_{\ell}x_{\ell})$. Since $y$ Pareto dominates $x$, it must be that $y_{\ell}P_{\ell}x_{\ell}$. But then $\ell\in U$, and $\ell$ was assigned before cycle $k$, a contradiction.

We now turn to the maximality claim. We show that for any domain $\tilde{\rcal}$ where $\tilde{\rcal}\nsubseteq\rcal(\hcal)$ for any partition $\hcal$ of $H$, there is some \ttc mechanism which is not PE on $\tilde{\rcal}$. If $\tilde{\rcal}\nsubseteq\mathcal{R}(\mathcal{H})$ for any $\hcal$, then $\tilde{\rcal}$ must contain two preference orderings, $R_{\alpha}$ and $R_{\beta}$, such that for some $h_{1},h_{2}\in H$ we have $h_{1}I_{\alpha}h_{2}$ but $h_{1}P_{\beta}h_{2}$. The proof seeks to isolate a pair like in Example \ref{ex:indiff-pe}. Taking only the existence of $R_{\alpha},R_{\beta}\in\tilde{\rcal}$ for granted, we find a preference profile $R\in\tilde{\rcal}^{n}$ and tie-breaking profile $\succ$ such that $\text{TTC}_{\succ}(R)$ is not PE. Without loss of generality, let $w_{i}=h_{i}$ for all $i\in N$. Let the preference profile $R$ be given by
\begin{align*}
R_{i}= & \begin{cases} R_{\alpha}\quad(\text{having } h_{1}I_{\alpha}h_{2}) & \text{if }i\in A := \{i\in N:w_{i}R_{\alpha}w_{1}\}\setminus\{2\}\\ R_{\beta}\quad(\text{having } h_{1}P_{\beta}h_{2}) & \text{if } i \in B := A^c.
\end{cases}
\end{align*}
Note that $R_{1}=R_{\alpha}$ and $R_{2}=R_{\beta}$. 

Take any tie-breaking profile $\succ$ such that for all $i\in N$, $i\succ_{i}j$ for all $j\neq i$. Denote $x=\text{TTC}_{\succ}(R)$. It follows from Lemma \ref{lem:xw} that $x=w$. However, note that $w_{2}I_{1}w_{1}$ and $w_{1}P_{2}w_{2}$, so $x$ is Pareto dominated by $y=(w_{2},w_{1},w_{3},...,w_{n})$.
\end{proof}
\medskip{}

\begin{proof}[Proof of (1) $\iff$ (3)]
First we show that for any objective indifferences domain, \ttc mechanisms are CS. Fix some tie-breaking rule $\succ$. Let $\mathcal{H}$ be any partition of $H$ and let $R\in\mathcal{R}(\mathcal{H})^{n}$. If $\mathcal{H}=\{H\}$, the result is trivial, so suppose the partition has at least two blocks. Suppose that the core of $(N,H,w,R)$ is non-empty and contains some allocation $y$. Let $x=\text{TTC}_{\succ}(R)$. It suffices to show that $x_{i}I_{i}y_{i}$ for all $i\in N$. We proceed by induction on the steps of $\text{TTC}_{\succ}(R)$.
\begin{adjustwidth}{1cm}{}
    \begin{enumerate}
        \item[Step 1.]  By definition of $\text{TTC}_{\succ}$, for all $i\in S_{1}$ we have $x_{i}R_{i}h$ for all $h\in H$; that is, agents in the first cycle receive (one of) their favorite houses. Therefore $x_{i}R_{i}y_{i}$ for all $i\in S_{1}$. Suppose there is some $j\in S_{1}$ such that $x_{j}P_{j}y_{j}$. Then $S_{1}$ and $x$ block $y$, contradicting the assumption that $y$ is in the core. Thus, $x_{i}I_{i}y_{i}$ for all $i\in S_{1}$.

        \item[Step $k$.]  Assume that $x_{i}I_{i}y_{i}$ for all $i\in\bar{S}_{k-1}$. Suppose that $y_{j}P_{j}x_{j}$ for some $j\in S_{k}$. By construction of objective indifferences, $\eta(y_{j})\neq\eta(x_{j})$. By Fact \ref{fact:assignedbefore}, $\{i\in N:x_{i}\in\eta(y_{j})\}\subseteq\bar{S}_{k-1}$. Because $\eta(y_{j})$ is finite, if $\eta(y_{j})\neq\eta(x_{j})$, there must be an agent $\ell$ in $\bar{S}_{k-1}$ for whom $x_{\ell}\in\eta(y_{j})$ but $y_{\ell}\notin\eta(y_{j})$. Therefore $\neg(y_{\ell}I_{\ell}x_{\ell})$, contradicting the assumption that $x_{i}I_{i}y_{i}$ for all $i\in\bar{S}_{k-1}$. Thus $x_{i}R_{i}y_{i}$ for all $i\in S_{k}$. Now suppose there is some $j\in S_{k}$ such that $x_{j}P_{j}y_{j}$. Then $S_{k}$ and $x$ block $y$, contradicting that $y$ is in the core.
    \end{enumerate}
\end{adjustwidth}
\noindent Thus $x_{i}I_{i}y_{i}$ for all $i\in N$, so $x$ must also be in the core as desired. (Since $y$ was an arbitrary allocation in the core, this also proves Corollary \ref{cor:uniquecore}.)

\vspace{0.5em}

We now turn to the maximality claim. We show that for any domain $\tilde{\rcal}$ where $\tilde{\rcal}\nsubseteq\rcal(\hcal)$ for some partition $\hcal$, there is some \ttc which is not CS on $\tilde{\rcal}$. As previously, if $\tilde{\rcal}\nsubseteq\mathcal{R}(\mathcal{H})$ for any $\hcal$, then $\tilde{\rcal}$ must contain two orderings, $R_{\alpha}$ and $R_{\beta}$, such that for some $h_{1},h_{2}\in H$ we have $h_{1}I_{\alpha}h_{2}$ but $h_{1}P_{\beta}h_{2}$. Without loss of generality, let $h_{1}$ be $R_{\beta}$'s highest ranked house such that $h_{1}I_{\alpha}h_{2}$ and $h_{1}P_{\beta}h_{2}$ are true. Formally, 
\begin{center}
    $h_{1}R_{\beta}h$ for all $h\in H$ such that $hI_{\alpha}h_{2}$. \hspace{3em} $(*)$
\end{center}
\noindent Also without loss of generality, let $w_{i}=h_{i}$ for all $i\in N$.

The following construction is the same as in the previous part. Define $A=\{i\in N:w_{i}R_{\alpha}w_{1}\}\setminus\{2\}$ and consider the preference profile where $R_{i}=R_{\alpha}$ for all $i\in A$ and $R_{i}=R_{\beta}$ for all $i\in B := A^{c}$. Take any tie-breaking profile $\succ$ such that for all $i\in N$, $i\succ_{i}j$ for all $j\neq i$.  Denote $x=\text{TTC}_{\succ}(R)$. It follows from Lemma \ref{lem:xw} that $x=w$. As before, $x$ is Pareto dominated by $y=(w_{2},w_{1},w_{3},...,w_{n})$, so $x$ is not in the core ($x$ is blocked by the grand coalition $N$ and $y$).

It remains to show that $y$ is in the core. Toward a contradiction, suppose there is a coalition $Q$ and allocation $z$ that blocks $y$. Let $U=\{i\in Q:z_{i}P_{i}y_{i}\}$ be the strict improvers under $z$, which must be non-empty. We proceed in steps denoted by the following claims.
\begin{claim}
\label{claim:wsubsetac} $U\cap A=\emptyset$.
\end{claim}
\begin{adjustwidth}{1cm}{}

Intuitively, if any agent in $A$ strictly prefers a house to his assignment, this house is also in $w_{A}$, preventing mutual improvements within $A$. Toward a contradiction, suppose $U_{A}:=U\cap A$ is non-empty; note that all agents in $U_{A}$ have the preference ranking $R_{\alpha}$. Take $i\in U_{A}$ such that $z_{i}R_{\alpha}z_{i'}$ for all $i'\in U_{A}$. That is, $i$ is the best off improver in $U_{A}$ (or one of the best off improvers). Observe that $z_{i}P_{\alpha}w_{i}$ since $i\in U$ and $w_{i}R_{\alpha}w_{1}$ since $i\in A$. Thus by construction of $R_{\alpha}$ and $A$, we have $J:=\left\{ j\in Q:w_{j}I_{\alpha}z_{i}\right\} \subseteq A$. That is, objects equivalent to $z_{i}$ are also owned by members of $A$. Also, $J$ is nonempty since $z_{i} \in w_{J}$. Further, note $1\not\in J$ since $z_{i}P_{\alpha}w_{1}$. Agent 1 is the only agent in $A$ who does not receive $y_i=w_i$, so $y_{j}=w_{j}$ for all $j\in J$.

Since $w_{J}$ is finite and $i\not\in J$, in order for $i$ to receive $z_{i}\in w_{J}$, some $j\in J$ must have $z_{j}\not\in w_{J}$. But then since $j\in Q\cap A$, it must be that $z_{j}P_{\alpha}w_{j}(=y_{j})$. Then $j \in U_A$, contradicting the assumption that $z_{i}R_{\alpha}z_{i'}$ for all $i'\in U_{A}$.

\end{adjustwidth}
\begin{claim}
\label{claim:aintq} $Q\cap A\neq\emptyset$.
\end{claim}
\begin{adjustwidth}{1cm}{}

We show that the coalition $Q$ cannot generate improvements with $w_{Q}$ if all its members have the same preference $R_{\beta}$. Toward a contradiction, suppose $Q\subseteq A^{c}$; then $R_{i}=R_{\beta}$ for all $i\in Q$. Let $i\in U$ be such that $z_{i}R_{\beta}z_{i'}$ for all $i'\in U$; that is, $i$ is the best off improver.

Now consider $J:=\{j\in Q:w_{j}I_{\beta}z_{i}\}\subseteq A^{c}$. Since $w_{J}$ is finite and $i\not\in J$, in order for $i$ to receive $z_{i}\in w_{J},$ some $j\in J$ must have $z_{j}\not\in w_{J}$. Then since $j\in Q$, it must be $z_{j}P_{\beta}w_{j}$. If $j \neq 2$, then $y_{j}=w_{j}$ so $z_{j}P_{\beta}y_{j}$. But then $j \in U$, contradicting $z_{i}R_{\beta}z_{i'}$ for all $i'\in U$.

Now suppose $j=2$. Since $w_{1}\in w_{A}$, we have $z_{2}\not=w_{1}(=y_{2})$. Then $z_{2}R_{\beta}w_{1}$ and $z_{2}\not=w_{1}$. That is, agent 2 must receive a house besides $w_1$ endowed to $Q\subseteq A^{c}$. Consider $L:=\{\ell\in Q:w_{\ell}R_{\beta}w_{1}\}\subseteq A^{c}$, and note that $2\not\in L$ since $w_{1}P_{\beta}w_{2}$. Since $w_{L}$ is finite, in order for $z_{2}\in w_{L}$, some $\ell\in L$ must receive $z_{\ell}\not\in w_{L}$. Then it must be that $w_{\ell}P_{\beta}z_{\ell}$. Since $\ell\in A^{c}$ and $\ell\not=2$, we have $y_{\ell}=w_{\ell}$. This gives $y_{\ell}P_{\beta}z_{\ell}$, contradicting that every agent in $Q$ weakly improves under $z$.

\end{adjustwidth}
\begin{claim}
Claims \ref{claim:wsubsetac} and \ref{claim:aintq} cannot both be true.
\end{claim}
\begin{adjustwidth}{1cm}{}

Without loss of generality, assume $Q$ forms a single trading cycle under $z$.\footnote{If the agents in $Q$ formed two or more trading cycles, then some cycle $S$ must contain an agent $i$ such that $z_{i}P_{i}y_{i}$, and this cycle $S$ suffices to form a blocking coalition.} Since $Q$ contains agents in both $A$ and $A^{c}$ in the same cycle, there exists some agent $\bar{b}\in A^{c}$ such that $z_{\bar{b}}\in w_{A}$ and some agent $\bar{a}\in A$ such that $z_{\bar{a}}\in w_{A^{c}}$. It must be that $z_{\bar{a}}=w_{2}$, since 
$z_{\bar a} I_\alpha w_{\bar a}$ and by construction $w_{\bar{a}}P_{\alpha}w_{i}$ for any $i\in A^{c}\backslash\{2\}$. For the same reason, $\bar{a}$ must be the only agent in $A\cap Q$ who receives a house from $w_{A^{c}}$.

Since only one agent in $A\cap Q$ receives a house from $w_{A^{c}}$, only one agent from $A^{c}\cap Q$ receives a house from $w_{A}$. We can represent the trading cycle $Q$ as \[ a_{1}\rightarrow a_{2}\rightarrow...\rightarrow a_{m}\rightarrow\bar{a}\rightarrow2\rightarrow b_{1}\rightarrow...\rightarrow\bar{b}\rightarrow a_{1} \] where $a_{1},...,a_{m},\bar{a}\in A$ and $2,b_{1},...,\bar{b}\in A^{c}$, and $z_{a_{1}}=w_{a_{2}}$, etc.

Since $U\cap A=\emptyset$, no $a\in\{a_{1},...,a_{m},\bar{a}\}$ strictly improves under $z$; that is, $z_{a}I_{\alpha}y_{a}$ for each $a$. Recall that $y_{i}I_{\alpha}w_{i}$ for all $i\in A$, since either 1) $y_{i}=w_{i}$ or 2) $i=1$, $y_{1}=w_{2}$ and $w_{2}I_{1}w_{1}$. Thus $w_{a_{1}}I_{\alpha}w_{a_{2}}I_{\alpha}\cdots I_{\alpha}w_{2}$.

So $z_{b}R_{\beta}y_{b}$ for all $b\in\{2,b_{1},...,\bar{b}\}$, with a strict relation for at least one $b$. Since $y_{2}=w_{1}$, we have $w_{b_{1}}R_{\beta}w_{1}$. For all other $b$, we have $y_{b}=w_{b}$, giving us $w_{a_{1}}R_{\beta}w_{\bar{b}}$, ..., $w_{b_{2}}R_{\beta}w_{b_{1}}$. At least one of these relations is strict; combining these gives $w_{a_{1}}P_{\beta}w_{1}$.

Then $w_{a_{1}}I_{\alpha}w_{2}$ and $w_{a_{1}}P_{\beta}w_{1}$. But this contradicts the assumption $(*)$.

\end{adjustwidth}

Thus we have that $y$ is in the core, as desired.
\end{proof}
\vspace{1em}

\ttc always selects an allocation in the weak core. This is a known property, but we give a short proof for convenience.

\noindent\textbf{Proposition \ref{prop:weakcore}}. \emph{For any
market, the weak core is non-empty and \ttc mechanisms select an allocation in the weak core.}
\begin{proof}
Fix a market $(N,H,w,R)$ and a tie-breaking profile $\succ$. Denote $x=\text{TTC}_{\succ}(R)\equiv \text{TTC}(R_{\succ})$. Toward a contradiction, suppose there exists a coalition $Q\subseteq N$ and allocation $y$ that weakly blocks $x$. That is, $y_{Q}=w_{Q}$ and $y_{i}P_{i}x_{i}$ for all $i\in Q$. Then we also have $y_{i}P_{i,\succ}x_{i}$ for all $i\in Q$, since the tie-breaking procedure preserves strict comparisons. Thus $y$ and $Q$ form a weak blocking coalition against $x$ \emph{according to the strict preferences} $R_{\succ}$. This violates the core property of TTC with strict preferences; thus \ttc always selects a mechanism in the weak core.

Since \ttc always produces an allocation, this also proves the existence of a weak core allocation.
\end{proof}

\subsection{Group strategy-proofness}

\label{gspappendix}

\noindent\textbf{Theorem \ref{thm:gsp}.} \emph{$\rcal$ is a symmetric-maximal
domain on which \ttc mechanisms are group strategy-proof if and only if it is an objective indifferences domain.}

\vspace{0.5em}

Before proving Theorem \ref{thm:gsp}, we review an important property of \ttc and state a useful lemma. Let $L(h,R_{i})=\{h'\in H:hR_{i}h'\}$ be the lower contour set of a preference ranking $R_{i}$ at house $h$.

\vspace{0.5em}

\noindent\textbf{Monotonicity (MON)}. A mechanism $f$ is \textbf{monotone}
if $f(R)=f(R')$ for any preference profiles $R$ and $R'$ such that $L(f_{i}(R),R_{i})\subseteq L(f_{i}(R),R_{i}')$ for all $i\in N$.
\vspace{0.5em}

\noindent That is, a mechanism $f$ is monotone if, whenever any set of agents move their allocations upwards in their preference rankings, the allocation remains the same. It is straightforward to show that TTC is monotone for strict preferences (e.g., \citet{Takamiya01}). Since $\text{TTC}_{\succ}(R)\equiv \text{TTC}(R_{\succ})$ for any $R$ and $\succ$, it follows directly that \ttc mechanisms are monotone.

The following result is an immediate consequence of Lemma 1 from \citet{ST24} and the fact that $\text{TTC}_{\succ}(R) \equiv \text{TTC}(R_\succ)$.
\begin{lem}[\citet{ST24}]
\label{lem:st_lem1}For any $R,R'$, let $x=\text{TTC}_{\succ}(R)$ and $y=\text{TTC}_{\succ}(R')$. Suppose there is some $i$ such that $y_{i}P_{i,\succ}x_{i}$. Then there exists some agent $j$ and house $h$ such that $x_{j}P_{j,\succ}h$ and $hP'_{j,\succ}x_{j}$.
\end{lem}
In other words, if $i$ is to receive a better house, some other agent $j$ must misreport his preferences and rank a worse house above his original allocation.

We now provide our proof of Theorem \ref{thm:gsp}.
\begin{proof}[Proof of Theorem \ref{thm:gsp}]
Fix $(N,H,w)$. The result is trivial for $|N|=1$, so assume $|N|\geq2$. First we show that for any objective indifferences domain, \ttc mechanisms are GSP. Fix some tie-breaking rule $\succ$. Let $\mathcal{H}$ be any partition of $H$ and let $R\in\mathcal{R}(\mathcal{H})^{n}$. If $\mathcal{H}=\{H\}$, the result is trivial, so assume that $\hcal$ has at least two blocks. Suppose $Q\subseteq N$ reports $R'_{Q}$ where $R'=(R'_{Q},R_{-Q})\in\rcal(\hcal)^{n}$. Let $y=\text{TTC}_{\succ}(R')$. We will show that if $y_{i}P_{i}x_{i}$ for some $i\in Q$, then $x_{j}P_{j}y_{j}$ for some $j\in Q$.

Let $R''=(R''_{Q},R_{-Q})$ be the preference profile in $\rcal(\hcal)^{n}$ such that for each $i\in Q$, $R''_{i}$ top-ranks the houses in $\eta(y_{i})$ and otherwise preserves the ordering of $R_{i}$. That is, for any $h\in\eta(y_{i})$ and $h'\notin\eta(y_{i})$, $hP''_{i}h'$; otherwise, $hR''_{i}h'$ if and only if $hR_{i}h'$. Let $z=\text{TTC}_{\succ}(R'')$. By monotonicity of $\text{TTC}_{\succ}$, $z=y$. Fix $i\in Q$ such that $y_{i}P_{i}x_{i}$. Since $z=y$, we have $z_{i}P_{i}x_{i}$ and consequently $z_{i}P_{i,\succ}x_{i}$. Applying Lemma \ref{lem:st_lem1}, there must be some $j\in Q$ and $h\in H$ such that $x_{j}P_{j,\succ}h$ but $hP''_{j,\succ}x_{j}$. Note that $h\notin\eta(x_{j})$; if it were, then for any $R,R''\in\rcal(\hcal)^{n}$, $x_{j}P_{j,\succ}h$ if and only if $x_{j}P''_{j,\succ}h$. Therefore, $x_{j}P_{j}h$ and $hP''_{j}x_{j}$.\footnote{This is where the restriction to objective indifferences is used. Under general indifferences, this is not necessarily true.} The only change from $R_{j}$ to $R''_{j}$ is to top-rank the houses in $\eta(y_{j})$, so it must be that $h\in\eta(y_{j})$. But then $x_{j}P_{j}y_{j}$, as desired.

\vspace{0.5em}

Next we show that for any symmetric domain $\tilde{\rcal}$ where $\tilde{\rcal}\nsubseteq\rcal(\hcal)$ for any partition $\hcal$ of $H$, \ttc mechanisms are not GSP on $\tilde{\rcal}$. As before, if $\tilde{\rcal}\not\subseteq\mathcal{R}(\mathcal{H})$ for any $\hcal$, then $\tilde{\rcal}$ must contain two orderings, $R_{\alpha}$ and $R_{\beta}$, such that for some $h_{1},h_{2}\in H$ we have $h_{1}I_{\alpha}h_{2}$ but $h_{1}P_{\beta}h_{2}$. The symmetric requirement also necessitates that $\tilde{\rcal}$ contains some $R_{\gamma}$ such that $h_{2}P_{\gamma}h_{1}$. Taking only the existence of $R_{\alpha},R_{\beta},R_{\gamma}\in\tilde{\rcal}$ for granted, we find a preference profile $R\in\tilde{\rcal}^{n}$ and tie-breaking profile $\succ$ such that $\text{TTC}_{\succ}(R)$ is not GSP.

Without loss of generality, assume $w_{i}=h_{i}$ for all $i\in N$. Consider the preference profile $R\in\tilde{\rcal}^{n}$ where
\begin{align*}
R_{i}= & \begin{cases} R_{\alpha}\quad(\text{having } w_{1}I_{\alpha}w_{2}) & \text{if } i \in A :=  \{i \in N : w_i R_\alpha w_1\} \setminus \{2\}
\\
R_{\beta} \quad (\text{having } w_{1}P_{\beta}w_{2}) & \text{if } i \in  B := \{i \in A^c : w_i R_\beta w_1\} \cup \{2\}\\
R_{\gamma} \quad (\text{having } w_{2}P_{\gamma}w_{1}) & \text{if } i \in C := (A \cup B)^c.
\end{cases}
\end{align*}

Note that $R_{1}=R_{\alpha}$ and $R_{2}=R_{\beta}$. Let $\succ$ be any tie-breaking profile such that for all $i\in N$, $i\succ_{i}j$ for all $j\neq i$. Additionally, let agent 2's tie-breaking order be $2\succ_{2}1\succ_{2}\cdots$. Let $x=\text{TTC}_{\succ}(R)$. Recall that by individual rationality of \ttc and our choice of $\succ$, for any agent $i$ if $x_i \neq w_i$ then $x_i P_i w_i$.  

\begin{claim}
\label{claim:rgsp} 
$x_{1}=w_{1}$ and $w_1 P_2 x_{2}$.
\end{claim}

\begin{adjustwidth}{1cm}{}
    \begin{proof}
        Let $\bar A = \{i \in N : w_i R_\alpha w_1\} \setminus \{2\}$ and $\bar B = \{i \in \bar A^c : w_i R_\beta w_1\}$. Note that $R_i = R_\alpha$ for all $i \in \bar A$ and $R_i = R_\beta$ for all $i \in \bar B$. It follows from Lemma \ref{lem:xw} (with $w_1$ filling the roles of both $h_1$ and $h_2$) that $x_i = w_i$ for all $i \in \bar A \cup \bar B$.  Therefore $x_1 = w_1$. Moreover since $R_2 = R_\beta$, if $x_2 R_2 w_1$ then $x_2 = w_i$ for some $i \in \bar A \cup \bar B$, a contradiction.  So $w_1 P_2 x_{2}$.
    \end{proof}
\end{adjustwidth}

\noindent Now suppose that agent 1 misreports $R_{1}'=R_{\gamma}$.
Let $R'=(R'_{1},R_{-1})$ and $y=\text{TTC}_{\succ}(R')$.
\begin{claim}
\label{claim:rprimegsp} $y_{1}=w_{2}$ and $y_{2}=w_{1}$.
\end{claim}
\begin{adjustwidth}{1cm}{}
\begin{proof}
    Let $\bar A = \{i \in N : w_i R_\alpha w_1\} \setminus \{1,2\}$ and $\bar B = \{i \in \bar A^c : w_i R_\beta w_1\} \setminus \{1\}$. Note that $R'_i = R_\alpha$ for all $i \in \bar A$ and $R'_i = R_\beta$ for all $i \in \bar B$. It follows from Lemma \ref{lem:xw} that $y_i = w_i$ for all $i \in \bar A \cup \bar B$.

    Let $\bar C = \{i \in (\bar A \cup \bar B)^c : w_i R_\gamma w_2\} \setminus \{2\}$. We will show that $y_{i}=w_{i}$ for all $i \in \bar C$. Note that $R'_i = R_\gamma$ for all $i \in \bar C$. Toward a contradiction, suppose $U:=\{i\in \bar C:y_{i}\neq w_{i}\}=\{i\in \bar C:y_{i}P_{\gamma}w_{i}\}$ is non-empty. Fix $i\in U$ such that $w_{i}R_{\gamma}w_{i'}$ for all $i'\in U$. Consider $j$ endowed with $y_{i}$; that is, $w_{j}=y_{i}$. Since $y_{j}\neq w_{j}$, it must be $j\notin \bar A \cup \bar B$, since $y_{i}=w_{i}$ for all $i\in \bar A \cup \bar B$. Then, since $y_{j} \neq w_{j}$ and $w_j P_\gamma w_i R_\gamma w_2$, it must be that $j \in U$. But this contradicts the assumption that $w_{i}R_{\gamma}w_{i'}$ for all $i'\in U$.

    Thus we have shown that $y_i = w_i$ for all $i \in \bar A \cup \bar B \cup \bar C$. Therefore, $y_i = w_i$ for all $i \neq 1$ such that $w_i R_\beta w_1$ and for all $i \neq 2$ such that $w_i R_\gamma w_2$. Consequently, $w_2 R_1' y_1$ (with indifference only if $y_1=w_2$) and $w_1 R_2' y_2$ (with indifference only if $y_2=w_1$).
    
    
    Toward a contradiction, suppose that $y_1 \neq w_2$. Then $w_2 P'_1 y_1$. By Fact \ref{fact:assignedbefore}, this implies that agent 2 was assigned before agent 1 was assigned during the process of $\text{TTC}_{\succ}(R')$. By definition of $\text{TTC}_\succ$, $x_2 R_2 w_1$. Then $y_2 = w_i$ for some $i \neq 1$, and $w_i R_\beta w_1$. But this contradicts $x_i = w_i$ for all $i \in \bar A \cup \bar B \cup \bar C$. By an analogous argument, it is straightforward to show that $y_2 = w_1$.
    
\end{proof}
\end{adjustwidth}

Observe that $w_{2}I_{1}w_{1}$ and $w_{1}P_{2}w_{2}$. Thus, via the misrepresentation $R'$, we have $y_{1}I_{1}x_{1}$ but $y_{2}P_{2}x_{2}$, showing $\text{TTC}_{\succ}(R)$ is not GSP.
\end{proof}

\section{Matching with capacities and priorities}

\label{schoolchoiceappendix}

We briefly note that TTC in the objective indifferences setting is not identical to TTC where the objects have capacities and priorities. Intuitively, in Shapley-Scarf markets with objective indifferences, the fixed tie-breaking rule determines which object an agent points at. Conversely, a priority ranking determines which agent an object points at. Consider the following example.
\begin{example}
Let the set of schools be $H=\{A,B,C\}$, with $C$ having 2 seats. Let the students be $N=\{a,b,c_{1},c_{2}\}$, where $a$ is ``endowed'' with $A$, and so on. Below are a possible priority ranking for the school choice setting and a possible tie-breaking profile for the Shapley-Scarf setting.

\begin{figure}[H]
\centering
\begin{subfigure}[t]{0.47\textwidth}
    \centering
\centering
\begin{tabular}{ccc}
$A$ & $B$ & $C$ \\
\hline 
$a$ & $b$ & $c_{1}$ \\
$b$ & $a$ & $c_{2}$ \\
$c_{1}$ & $c_{2}$ & $a$ \\
$c_{1}$ & $c_{1}$ & $b$
\end{tabular}
    \caption*{\small School priority for the school choice setting}
\end{subfigure}
\hfill
\begin{subfigure}[t]{0.47\textwidth}
    \centering
\centering
\begin{tabular}{cccc}
$\succ_a$ & $\succ_b$ & $\succ_{c_{1}}$ & $\succ_{c_{2}}$ \\
\hline 
$c_{1}$ & $c_{2}$ & $c_{1}$ & $c_{1}$ \\
$c_{2}$ & $c_{1}$ & $c_{2}$ & $c_{2}$ \\
$a$ & $a$ & $a$ & $a$ \\
$b$ & $b$ & $b$ & $b$
\end{tabular}
    \caption*{\small Tie-breaking profile for Shapley-Scarf setting}
\end{subfigure}
\end{figure}

\noindent Consider the preference profiles $R=(R_{a},R_{b},R_{c_{1}},R_{c_{2}})$
and $R'=(R'_{a},R'_{b},R'_{c_{1}},R'_{c_{2}})$, shown below.

\begin{figure}[H]
\centering
\begin{subfigure}[t]{0.47\textwidth}
    \centering
\centering
\begin{tabular}{cccc}
$R_a$ & $R_b$ & $R_{c_{1}}$ & $R_{c_{2}}$ \\
\hline 
$C$ & $C$ & $A$ & $A$ \\
$A$ & $A$ & $B$ & $B$ \\
$B$ & $B$ & $C$ & $C$
\end{tabular}
    \caption*{\small Preference profile $R$}
\end{subfigure}
\hfill
\begin{subfigure}[t]{0.47\textwidth}
    \centering
\centering
\begin{tabular}{cccc}
$R'_a$ & $R'_b$ & $R'_{c_{1}}$ & $R'_{c_{2}}$ \\
\hline 
$C$ & $C$ & $B$ & $B$ \\
$A$ & $A$ & $A$ & $A$ \\
$B$ & $B$ & $C$ & $C$
\end{tabular}
    \caption*{\small Preference profile $R'$}
\end{subfigure}
\end{figure}

Note that under $R$ and $R'$, both $c_{1}$ and $c_{2}$ have the same preferences. TTC with school priorities results in $(A:c_{1},B:c_{2},C:ab)$ and $(A:c_{2},B:c_{1},C:ab)$ under $R$ and $R'$ respectively. Note that $c_{1}$ must receive his preferred school in either case, since his priority at school $C$ is higher than $c_{2}$'s priority at school $C$. In fact, in the school choice setting, since $c_{1}$ has higher priority at $C$ than $c_{2}$ has, whenever $c_{1}$ and $c_{2}$ have the same preferences $c_{1}$ will weakly prefer his assignment to $c_{2}$'s assignment. By contrast, in the Shapley-Scarf setting, $\text{TTC}_{\succ}$ results in $(A:c_{1},B:c_{2},C:ab)$ and $(A:c_{1},B:c_{2},C:ab)$ under $R$ and $R'$ respectively. Now, $c_{1}$ does not necessarily receive his top choice when $c_{1}$ and $c_{2}$ have the same preferences. Under $R'$, $c_{2}$ gets his top choice at the expense of $c_{1}$, since $c_{2}\succ_{b}c_{1}$.

It is not always possible to recreate a school priority order with a tie-breaking profile, and vice versa.
\end{example}

\end{document}